\documentclass[11pt,letterpaper]{article}
\usepackage{color,ifpdf,latexsym,graphicx,url}
\title{Continuous Blooming of Convex Polyhedra}

\author{%
  Erik D. Demaine%
    \thanks{MIT Computer Science and Artificial Intelligence Laboratory,
      32 Vassar St., Cambridge, MA 02139, USA,
      \protect\url{{edemaine,mdemaine}@mit.edu}}
    \thanks{Partially supported by NSF CAREER award CCF-0347776.}
\and
  Martin L. Demaine\footnotemark[1]
\and
  Vi Hart%
    \thanks{\protect\url{http://vihart.com}}
\and
  John Iacono%
    \thanks{Department of Computer Science and Engineering,
      Polytechnic Institute of NYU, Brooklyn, NY, USA.
      \protect\url{http://john.poly.edu}}
\and
  Stefan Langerman%
    \thanks{Ma\^itre de recherches du FNRS, D\'epartment d'Informatique,
      Universit\'e Libre de Bruxelles, Brussels, Belgium,
      \protect\url{stefan.langerman@ulb.ac.be}}
\and
  Joseph O'Rourke%
    \thanks{Department of Computer Science, Smith College,
        Northampton, MA 01063, USA, \protect\url{orourke@cs.smith.edu}}
}

\usepackage
  [breaklinks,bookmarks,bookmarksnumbered,bookmarksopen,bookmarksopenlevel=2]
  {hyperref}
{\makeatletter \hypersetup{pdftitle={\@title}}}
\hypersetup{pdfauthor={ Erik D. Demaine,  Martin L.  Demaine,
    John Iacono, Stefan Langerman,  Joseph O'Rourke}}

\advance \topmargin by -\headheight
\advance \topmargin by -\headsep
\evensidemargin \oddsidemargin
{\makeatletter \gdef\fps@figure{!htbp}}

\let\realbfseries=\bfseries
\def\bfseries{\realbfseries\boldmath}

\newtheorem{theorem}{Theorem}
\newtheorem{lemma}[theorem]{Lemma}
\newtheorem{claim}[theorem]{Claim}
\newtheorem{corollary}[theorem]{Corollary}

\newtheorem{algorithm}{Algorithm}
\newtheorem{open}{Open Problem}

\makeatletter
\def\GrabProofArgument[#1]{ #1: \egroup\ignorespaces}
\def\proof{\noindent\textbf\bgroup Proof%
           \@ifnextchar[{\GrabProofArgument}{: \egroup\ignorespaces}}

\makeatother

\let\epsilon=\varepsilon
\let\eps=\varepsilon

\begin{document}
\maketitle

\begin{abstract}
  We construct the first two continuous bloomings of all convex polyhedra.
  First, the source unfolding can be continuously bloomed.
  Second, any unfolding of a convex polyhedron can be refined (further cut,
  by a linear number of cuts) to have a continuous blooming.
\end{abstract}

\section{Introduction}

A standard approach to building 3D surfaces from rigid sheet material,
such as sheet metal or cardboard, is to design an \emph{unfolding}:
cuts on the 3D surface so that the remainder can unfold (along edge hinges)
into a single flat non-self-overlapping piece.
The advantage of this approach is that many (relatively cheap)
technologies---such as NC machining/milling, laser cutters, waterjet cutters,
and sign cutters---enable manufacture of an arbitrary flat shape
(with hinges) from a sheet of material.
As a result, existence of and algorithms for unfolding have been studied
extensively in the mathematical literature.

Often overlooked in this literature, however, is the second manufacturing step:
can we actually fold the flat shape along the hinges into the desired 3D
surface, without self-intersection throughout the motion?
Some unfoldings have no such motion \cite[Theorem~4]{Biedl-Lubiw-Sun-2005}.
In this paper, we develop such motions for large families of unfoldings of
convex polyhedra.  In particular, we establish for the first time that
every convex polyhedron can be manufactured by folding an unfolding.

\paragraph{Unfolding.}
More precisely, an \emph{unfolding} of a polyhedral surface in 3D consists of
a set of cuts (arcs) on the surface whose removal result in a surface
whose intrinsic metric is equivalent to the interior of a flat
non-self-overlapping polygon,
called the \emph{development} or \emph{unfolded shape} of the unfolding.
The cuts in an unfolding of a convex polyhedron form a tree,
necessarily spanning all vertices of the polyhedron (enabling them to flatten)
\cite{Bern-Demaine-Eppstein-Kuo-Mantler-Snoeyink-2003}.
All unfoldings we consider allow cuts anywhere on the polyhedral surface,
not just along edges.
Four general unfolding algorithms are known for arbitrary convex polyhedra:
the source unfolding \cite{Sharir-Schorr-1986,Miller-Pak-2008},
the star unfolding \cite{Aronov-O'Rourke-1992},
and two variations thereof \cite{Itoh-O'Rourke-Vilcu-2008,Itoh-O'Rourke-Vilcu-2009}.
Positive and negative results for unfolding nonconvex polyhedra can be found in
\cite{Bern-Demaine-Eppstein-Kuo-Mantler-Snoeyink-2003,Damian-Flatland-O'Rourke-2007-epsilon,O'Rourke-2008-orthosurvey}.

\paragraph{Blooming.}
Imagine the faces of the cut surface in an unfolding as rigid plates,
and the (sub)edges of the polyhedron connecting them as hinges.
A \emph{continuous blooming} of an unfolding is a continuous motion of this
plate-and-hinge structure from the original shape of the 3D polyhedron
(minus the cuts) to the flat shape of the development,
while avoiding intersection between the plates throughout the motion.
In 1999, Biedl, Lubiw, and Sun \cite{Biedl-Lubiw-Sun-2005} gave an
example of an unfolding of an orthogonal (nonconvex) polyhedron
that cannot be continuously bloomed.
In 2003, Miller and Pak \cite[Conjecture~9.12]{Miller-Pak-2008}
reported the conjecture of Connelly that every convex polyhedron has a
nonoverlapping unfolding that can be continuously bloomed.
He further conjectured that the blooming can monotonically open all
dihedral angles of the hinges.
More recently, Pak and Pinchasi \cite{Pak-Pinchasi-2009} describe a
simple blooming algorithm for convex polyhedra which they show works
for the finite class of Archimedean solids
(excluding prisms and antiprisms, by extending existing bloomings of Platonic solids)
but fails on other polyhedra.
No other nontrivial positive results have been established.

\paragraph{Our results.}
We prove Connelly's conjecture by giving the first two general algorithms
for continuous blooming of certain unfoldings of arbitrary convex polyhedra,
which also monotonically open all hinge dihedral angles.
Both of our algorithms have a relatively simple structure:
they perform a sequence of linearly many steps of the form
``open one dihedral angle uniformly by angle $\alpha$''.
Thus our algorithms open angles monotonically, uniformly (at constant speed),
and one at a time. 
The challenge in each case is to prove that the motions cause no intersection.

First we show in Section~\ref{Blooming a Refinement of Any Unfolding}
that every unfolding can be refined (further cut, by a linear number of cuts)
into another unfolding with a continuous blooming.  Indeed, we show that
any serpentine unfolding (whose dual tree is a path) has a continuous blooming,
and then standard techniques can refine any unfolding into a serpentine
unfolding.

Next we show in Section~\ref{Blooming the Source Unfolding}
that one particularly natural unfolding, the source unfolding,
has a continuous blooming.  It is in the context of the source unfolding
that Miller and Pak \cite{Miller-Pak-2008} described continuous blooming.

Finally we mention in Section~\ref{Open Problems} an unsolved more-general
form of continuous blooming for the source unfolding that resembles
Cauchy's arm lemma.

\section{Blooming a Refinement of Any Unfolding}
\label{Blooming a Refinement of Any Unfolding}

In this section, we show that any unfolding can be refined into
an unfolding with a continuous blooming.
An unfolding $U$ is a \emph{refinement} of another unfolding $U'$
if the cuts in $U$ form a superset of the cuts in~$U'$.
Our approach is to make the unfolding serpentine:
an unfolding is \emph{serpentine} if the dual tree of the faces
is in fact a path.
Here ``faces'' is a combinatorial notion, not necessarily coinciding
with geometry: we allow artificial edges with flat dihedral angles.
We do require, however, that every two faces adjacent in the dual path
form a nonflat dihedral angle; otherwise, those faces could be merged.
We use the term \emph{facet} when referring to the original
geometric facets of the polyhedron~$Q$.

\begin{lemma}
  Every unfolding can be refined by a linear number of additional cuts
  into a serpentine unfolding.
\end{lemma}

\begin{proof}
  The polyhedron and unfolding together define \emph{faces} bounded
  by polyhedron edges and by cuts in the unfolding.
  For every two such faces that share an uncut edge,
  we refine by adding a cut from the centroid of each face to the
  midpoint of the shared edge.
  Figure~\ref{CubeRefine} shows an example.
  These cuts form a tree isomorphic to the
  dual tree of the unfolding, with a node for each face and an arc connecting
  two faces that share an uncut edge in the original unfolding.
  Thus, after the refinement, the remaining subfaces can be connected together
  in a Hamiltonian cycle via an Euler tour around the tree of added cuts.
  (This powerful trick originates in
  \cite[Theorem~2.4]{Arkin-Held-Mitchell-Skiena-1996}
  for the case of triangulations.)
  Finally, we make one additional cut along a (sub)edge
  to divide the dual cycle into a dual path.
  If the original polyhedron and unfolding define $f$ faces,
  the total number of added cuts is $2 (f-1) + 1 = 2f-1$.

\begin{figure}
  \centering
  \includegraphics[scale=0.5]{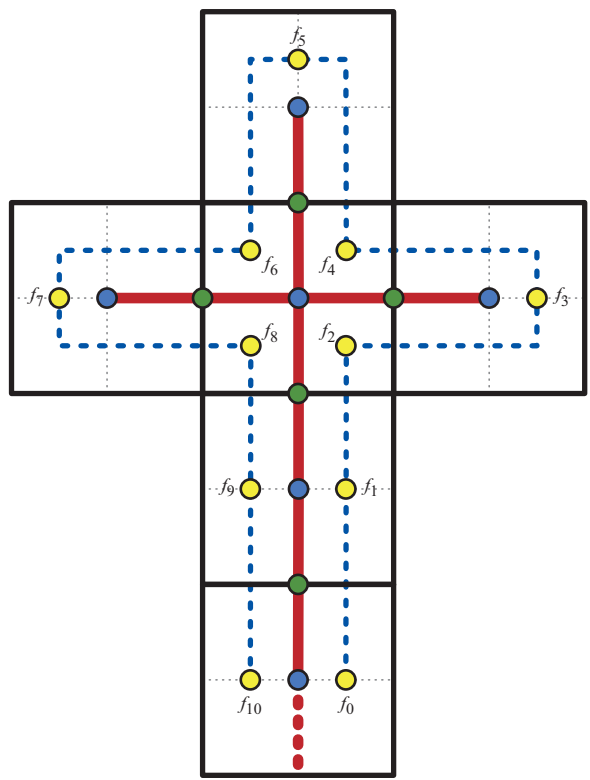}
  \caption{Refinement of the Latin cross unfolding of a cube
    to a serpentine unfolding.  The dashed cut at the base of the cross
    is the one additional cut that divides the cycle to a path.}
  \label{CubeRefine}
\end{figure}

  We observe that any refinement that preserves connectivity (such as this one)
  also preserves that the cuts form a valid unfolding, with a
  non-self-overlapping development.  The additional cuts in the refinement
  are intrinsic to the surface metric, so they can be applied to the unfolding
  just as well both on the 3D polyhedron and developed in the plane.
  Hence the refined unfolding develops in the plane to a region with exactly
  the same closure as the original unfolding, implying that the refined
  unfolding has no (interior) self-overlap if and only if the original
  unfolding has none.
\end{proof}

It remains to prove that such a refinement suffices:

\begin{theorem} \label{Path-Unroll theorem}
  Any serpentine unfolding can be continuously bloomed. 
\end{theorem}

\begin{corollary}
  Any unfolding of a polyhedron can be refined by a linear number of
  additional cuts into an unfolding with a continuous blooming.
\end{corollary}

Toward Theorem~\ref{Path-Unroll theorem},
we start with the following simple blooming algorithm.
Figure~\ref{CubeUnroll} shows an example of the algorithm in action.

\begin{algorithm}[Path-Unroll]
  Suppose we are given a serpentine unfolding of a polyhedron $Q$
  whose dual path is $P = \langle f_0, f_1, \dots, f_k \rangle$.
  For $i = 1, 2, \dots, k$ in sequence,
  uniformly open the dihedral angle between $f_{i-1}$ and $f_i$
  until those two faces become coplanar.
\end{algorithm}

\begin{figure}
  \centering
  \includegraphics{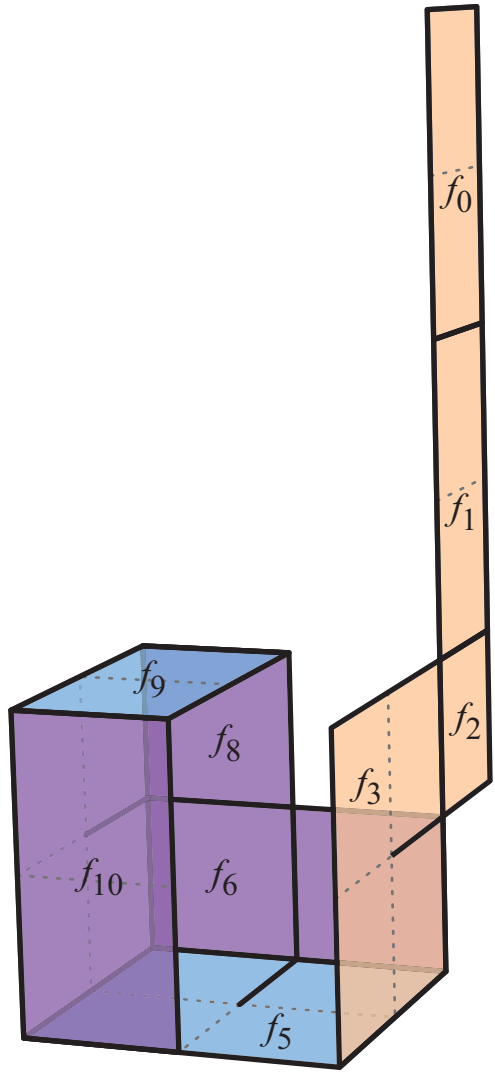}
  \caption{The Path-Unroll algorithm applied to the refined cube unfolding
    of Figure~\protect\ref{CubeRefine}, after the first $i=3$ steps.}
  \label{CubeUnroll}
\end{figure}

The Path-Unroll algorithm is almost what we need: we show below that
it only causes touching, not crossing, between faces that are coplanar
with a facet of $Q$ and touching between faces whose supporting planes
intersect in a line supporting an edge of~$Q$.
Precisely, two (convex) sets are \emph{noncrossing} if there is a plane $H$
such that both sets lie in opposite closed halfspaces bounded by~$H$,
while two sets \emph{touch} if they have a common intersection but do not cross.

\begin{lemma}\label{lem:path-unroll}
  The Path-Unroll algorithm causes no crossing between faces of the unfolding.
  Furthermore, (two-dimensional) touching between faces can occur only at the
  end of step $i$ between faces coplanar with the facet containing~$f_i$,
  and between the beginning and the end of step $i$ at the edge
  of $Q$ shared by $f_{i-1}$ and~$f_i$.
\end{lemma}

\begin{proof}
  The Path-Unroll algorithm has the invariant that, at step $i$ of the algorithm,
  the prefix $f_0, f_1, \dots, f_{i-1}$ of faces lies entirely in the plane
  containing $f_{i-1}$, and because this planar prefix appears as a subset
  of the final (nonoverlapping) unfolding, it does not self-intersect.
  The suffix $f_i, f_{i+1}, \dots, f_k$ of faces is a subset of the original
  polyhedral surface, so does not self-intersect.

  It remains to show that the prefix and suffix do not cross each other
  and intersections occur only as described in the statement of the lemma.
  At the beginning of step $i$ of the path unroll algorithm, the plane containing
  $f_{i-1}$ intersects the polyhedron at precisely the facet of $Q$
  containing $f_{i-1}$.  At the end
  of step~$i$, the plane intersects the polyhedron at precisely the
  facet of $Q$ containing $f_i$.
  In the middle of step~$i$, the plane intersects the polyhedron at
  precisely the edge of $Q$ bounding $f_{i-1}$ and~$f_i$;
  at all times, the plane remains tangent to the polyhedron.
  Thus the prefix $f_0, f_1, \dots, f_{i-1}$ only touches
  the suffix $f_i,f_{i+1},\ldots,f_k$ at a facet of $Q$ at the
  beginning and end of step $i$, and at an edge of $Q$ during step~$i$.
\end{proof}

To avoid touching at faces, we modify the algorithm slightly.

\begin{algorithm}[Path-TwoStep]
  Suppose we are given a serpentine unfolding of a polyhedron $Q$
  whose dual path is $P = \langle f_0, f_1, \dots, f_k \rangle$. Let
  $0<\phi_i<\pi$ be the dihedral angle between $f_{i-1}$ and $f_i$,
  let $\psi_i=\pi-\phi_i$ be the required folding angle,
  and let $\Psi_i=\sum_{j=1}^i \psi_j$ be the prefix sum for $i=1,2,\ldots,k$.
  Choose $\eps>0$ small enough. 
  \begin{itemize}
  \item At time $t\in[0,\Psi_1-\eps]$, uniformly open the dihedral angle between $f_0$
  and $f_1$ until they form a dihedral angle of $\pi-\eps$.
  \item For $i = 2, 3, \dots, k$:
    \begin{itemize}
    \item At time $t\in[\Psi_{i-1}-\eps,\Psi_i-2\eps]$,
      uniformly open the dihedral angle between $f_{i-1}$ and $f_i$ until they
      form a dihedral angle of $\pi-\eps$.
    \item At time $t\in[\Psi_i-2\eps,\Psi_i-\eps]$,
      uniformly open the dihedral angle between $f_{i-2}$ and $f_{i-1}$
      until those two faces become coplanar.
    \end{itemize}
  \item At time $t\in[\Psi_k-\eps,\Psi_k]$, uniformly open the dihedral angle between $f_{k-1}$ and $f_k$
  until those two faces become coplanar.
  \end{itemize}
\end{algorithm}

This algorithm avoids two-dimensional touching but may still cause
one-dimensional touching.

\begin{lemma}\label{lem:twostep}
  The Path-TwoStep algorithm causes no crossing between faces of the unfolding.
  Furthermore, touching can occur only at time $t=\Psi_i-\eps$ for $i=1,2,\dots,k$
  at the edge of $Q$ shared by $f_{i-1}$ and~$f_i$.
\end{lemma}

\begin{proof}
Let $\ell_j$ be the supporting line of the edge shared by $f_{j-1}$ and~$f_j$.
Let $q_j$ be the intersection point (in projective space) between the
lines $\ell_j$ and $\ell_{j-1}$, which lie in the common plane of~$f_{j-1}$.
At times $t\in[\Psi_{i-1}-\eps,\Psi_i-\eps]$, consider the central projection of $Q$ and the unfolding,
from point $q_i$ onto a plane perpendicular to the bisector of
$\ell_i$ and $\ell_{i-1}$, as shown in Figure~\ref{fig:TwoStep}.
(If $q_i$ is a point at infinity, this projection is an orthogonal projection
onto a plane perpendicular to the direction of~$q_i$.)
\begin{figure}
  \centering
  \includegraphics[scale=0.6]{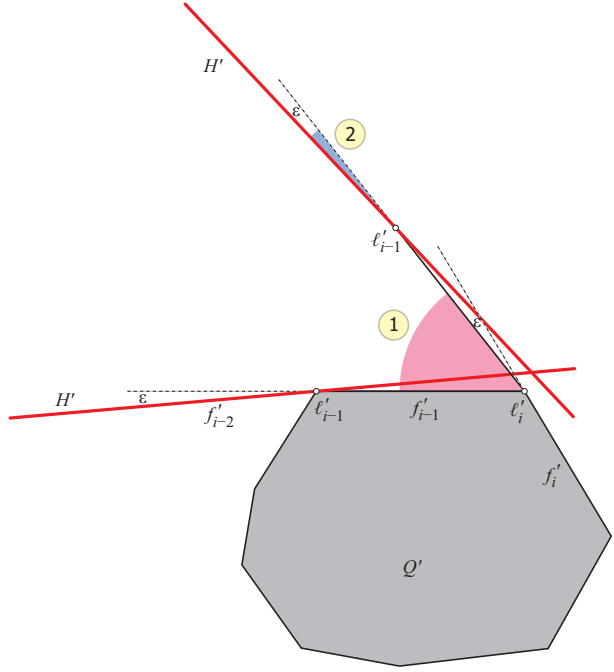}
  \caption{The central projection from point~$q_i$.
    During $t \in [\Psi_{i-1}-\eps,\Psi_i-2\eps]$, rotation (marked 1 in the figure) occurs in projection about $\ell'_i$.
    During $t \in [\Psi_i-2\eps,\Psi_i-\eps]$, rotation (marked 2) straightens
    the angle at the new position of~$\ell'_{i-1}$.
  }
  \label{fig:TwoStep}
\end{figure}
In that projection, $Q$ projects to a convex polygon $Q'$, and so the
projection of the suffix $f_i, f_{i+1}, \dots, f_k$ of untouched faces is
within that polygon.
Face $f_{i-1}$ projects to the segment $f'_{i-1}$ connecting the points
$\ell'_i$ and $\ell'_{i-1}$ representing the lines $\ell_i$ and~$\ell_{i-1}$.
Just like in the invariant of Lemma~\ref{lem:path-unroll}, the prefix
$f_0, f_1, \dots, f_{i-2}$ lies in the plane $H$ of $f_{i-2}$, which is
incident to $\ell_{i-1}$ and thus projects to a line~$H'$.

At time $t=\Psi_{i-1}-\eps$, the face $f_{i-1}$ is in its original
position and so does not intersect any of the faces
$f_i, f_{i+1}, \dots, f_k$. The projection of $f_0, f_1, \dots, f_{i-2}$ is
contained in line~$H'$, which is tangent to $Q'$ at point
$\ell_{i-1}'$, assuming $\eps<\psi_{i-1}$.
When $t\in[\Psi_{i-1}-\eps,\Psi_i-2\eps]$, as the angle between $f_i$ and $f_{i-1}$ opens,
the line containing edge $f_{i-1}'$ remains tangent to $Q'$ while
rotating about point $\ell_i'$. Thus $f_{i-1}$ touches $Q$ exactly
at its edge connecting it to $f_i$ and avoids any intersection.

During the motion, $H'$ remains tangent to
the convex hull of $Q'$ and $f_{i-1}'$ at the projection of
$\ell_{i-1}$, assuming $\eps$ is chosen so that
$\pi-\eps$ is larger than the dihedral angle of the edge in
$\ell_{i-1}$ of the convex hull of the suffix $f_{i-1}, f_i, \dots,
f_k$ when the dihedral angle between $f_i$ and $f_{i-1}$ is $\pi$.

At times $t\in[\Psi_i-2\eps,\Psi_i-\eps]$, as the dihedral angle between
$f_{i-2}$ and $f_{i-1}$ opens until those two faces become
coplanar, $H'$ remains tangent to the convex hull of $Q'$ and
$f_{i-1}'$ at the projection of $\ell_{i-1}$ and so does not intersect
$Q'$ until time $\Psi_i-\eps$, when $H'$ becomes collinear with
$f_{i-1}'$. At precisely that instant, $H'$ touches $Q'$ at the
projection of $\ell_i$, and so the faces in the prefix 
$f_0, f_1,\dots, f_{i-1}$ might touch, but not cross, edges still on
$Q$ at the edge between $f_{i-1}$ and~$f_i$.
\end{proof}

To avoid this one-dimensional touching, we modify the algorithm once more. 

\begin{algorithm}[Path-Waltz]
  Suppose we are given a serpentine unfolding of a polyhedron $Q$
  whose dual path is $P = \langle f_0, f_1, \dots, f_k \rangle$. Let
  $0<\phi_i<\pi$  be the dihedral angle between $f_{i-1}$ and $f_i$,
  let $\psi_i=\pi-\phi_i$ be the required folding angle,
  and let $\Psi_i=\sum_{j=1}^i \psi_j$ be the prefix sum for $i=1,2,\dots,k$.
  Choose $\eps>0$ and $\delta>0$ small enough. 
  \begin{itemize}
    \item At time $t\in[0,\Psi_1-\eps]$, uniformly open the dihedral angle between $f_0$ and $f_1$ until they
      form a dihedral angle of $\pi-\eps$. 
    \item At time $t\in[\Psi_1-\eps,\Psi_1-\eps+\delta]$, uniformly
      open the dihedral angle between $f_1$ and $f_2$ by an angle $\delta_2$.
    \item For $i = 2, 3, \dots, k-1$:
      \begin{itemize}
        \item At time $t\in[\Psi_{i-1}-\eps+\delta,\Psi_i-2\eps]$,
          uniformly open the dihedral angle between $f_{i-1}$ and $f_i$ until they
          form a dihedral angle of $\pi-\eps$.
        \item At time $t\in[\Psi_i-2\eps,\Psi_i-2\eps+\delta]$,
          uniformly open the dihedral angle between $f_i$ and
          $f_{i+1}$ by an angle $\delta$.
        \item At time $t\in[\Psi_i-2\eps+\delta,\Psi_i-\eps+\delta]$,
          uniformly open the dihedral angle between $f_{i-2}$ and $f_{i-1}$
          until those two faces become coplanar.
        \end{itemize}
      \item At time $t\in[\Psi_{k-1}-\eps+\delta, \Psi_k-\eps]$,
        uniformly open the dihedral angle between $f_{k-1}$
        and~$f_k$ until those two faces become coplanar.. 
      \item At time $t\in[\Psi_k-\eps, \Psi_k]$, uniformly open the
        dihedral angle between $f_{k-2}$ and $f_{k-1}$
        until those two faces become coplanar.
      \end{itemize}
\end{algorithm}

\begin{figure}
  \centering
  \includegraphics[scale=0.6]{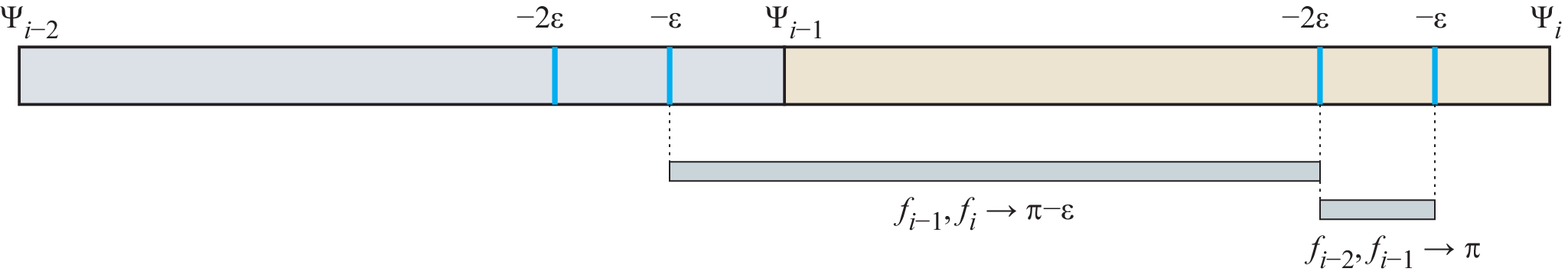}
  \includegraphics[scale=0.6]{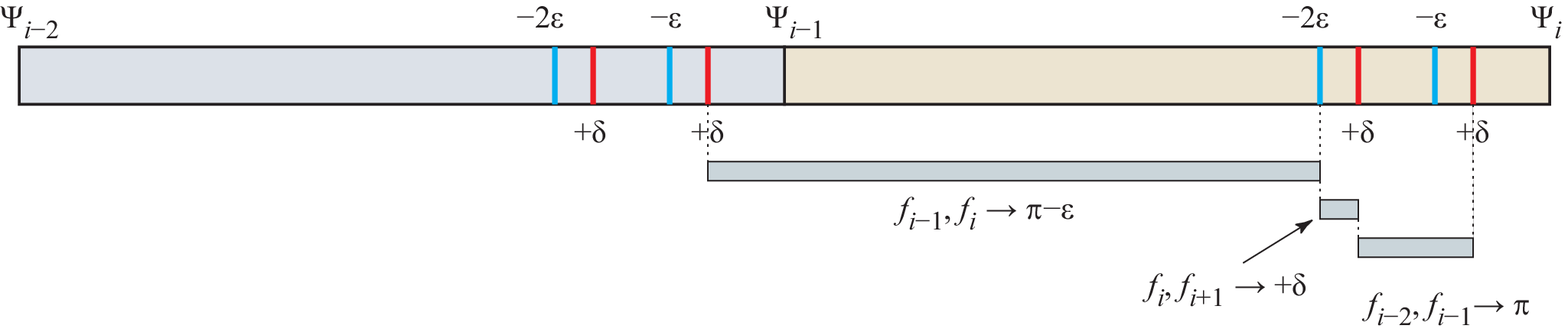}
  \caption{
    Timing diagrams for Path-TwoStep (top) and Path-Waltz (bottom).
  }
  \label{fig:timing}
\end{figure}

Figure~\ref{fig:timing} shows a timing diagram of the Path-TwoStep and
Path-Waltz algorithms.  Path-Waltz avoids all touching except that present
in the original unfolding, proving our desired theorem about serpentine
unfoldings having continuous bloomings:

\medskip

\begin{proof}[of Theorem~\ref{Path-Unroll theorem}]
  We prove that Path-Waltz causes no intersection. 
  Note that during the times $t\in[\Psi_{i-1}-\eps+\delta,\Psi_i-2\eps]$
  for some~$i$, the configuration is exactly the same as in the Path-TwoStep
  algorithm, and so no intersection occurs during those times. 

  Let $H_i$ be the plane containing face $f_i$. 
  At time $t=\Psi_i-2\eps$, the configuration occurs in
  Path-TwoStep, and thus does not self-intersect. 
  Furthermore, the proof of Lemma~\ref{lem:twostep} shows that the plane
  $H_{i-2}$ does not intersect $Q$ unless the edge shared by $f_i$ and
  $f_{i-1}$ has a common endpoint with the edge shared by $f_{i-1}$ and
  $f_{i-2}$ (and that common endpoint would be~$q_{i-1}$). Assume for
  now that this is not the case.   
  Because the non-self-intersecting configuration space is open,
  $\delta$ can be chosen small enough
  that, at times $t\in[\Psi_i-2\eps, \Psi_i-2\eps+\delta]$, 
  as the angle between $f_i$ and $f_{i+1}$ opens, no faces intersect,
  and $H_{i-2}$ does not intersect~$Q$.
  Note that during those times, the suffix $f_{i-1},f_i,\ldots,f_k$ is
  in a configuration identical to that in the Path-Twostep algorithm
  at times $t'\in[\Psi_i-\eps, \Psi_i-\eps+\delta]$ and again, by the
  proof of Lemma~\ref{lem:twostep}, the plane $H_{i-1}$ does not touch
  the interior of $Q$. Therefore, the open double-wedge between
  $H_{i-1}$ and $H_{i-2}$ remains empty throughout the motion. 
  When, at times 
  $t\in[\Psi_i-2\eps+\delta,\Psi_i-\eps+\delta)$,
  the angle between $f_{i-2}$ and $f_{i-1}$ opens, the faces in the
  suffix $f_{i-2},f_{i-1},\ldots,f_k$ remain in that open double-wedge
  and do not cause any intersection.

  It remains to show that no intersections occur during the motion if
  the edge shared by $f_i$ and $f_{i-1}$ and the edge shared by $f_{i-1}$ and
  $f_{i-2}$ both have $q_{i-1}$ as a common endpoint.
  For this case, notice that an intersection can occur only
  in a small neighborhood of point~$q_{i-1}$.
  The angle $\eps$ can be chosen small enough so
  that the normal vectors to faces $f_i$, $f_{i-1}$, and $f_{i-2}$ all
  have pairwise angles smaller than $\pi/2$, and so the opening does not
  cause intersections with the double wedge.
\end{proof}

\section{Blooming the Source Unfolding}
\label{Blooming the Source Unfolding}

This section proves that the following algorithm continuously blooms
the source unfolding of any convex polyhedron.
Figure~\ref{ChineseTakeOut} shows a simple example
of the algorithm in action.

\begin{algorithm}[Tree-Unroll]
  Let $T$ be the tree of faces formed by the source unfolding,
  rooted at the face containing the source $s$ of the unfolding.
  For each face in the order given by a post-order traversal of~$T$,
  uniformly open the dihedral angle between the face and its parent in $T$
  until the two faces become coplanar.
 \end{algorithm}

\begin{figure}
  \centering
  \includegraphics[scale=0.6]{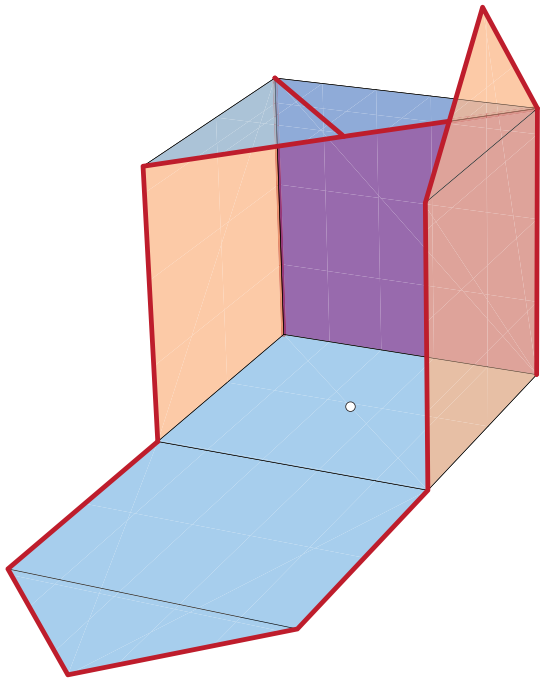}
  \caption{Start of Algorithm Tree-Unroll applied to the source unfolding of
    a cube (with cuts drawn in bold lines) with the source point in the center
    of the bottom face.}
  \label{ChineseTakeOut}
\end{figure}

First we need some tools for arguing about shortest paths and source
unfoldings.  Our first tool shows that source unfoldings change little
when ``growing'' a polyhedron:

\begin{lemma} \label{growing}
  If the interior of a convex polyhedron $Q$ is contained
  in the interior of a convex polyhedron~$Q'$,
  and $p$ and $q$ are points on the boundary of both $Q$ and~$Q'$,
  then the shortest path between $p$ and $q$ is at least as long
  on $Q'$ as on~$Q$.
\end{lemma}

\begin{proof}
  The shortest path in 3D between two points $p$ and $q$ on the boundary of a
  convex polyhedron~$Q$, while avoiding $Q$ as an obstacle, lies entirely on
  the surface of~$Q$: otherwise, each connected part of the shortest path
  not intersecting $Q$ would be a segment joining two points on the
  surface of~$Q$, but then by convexity of~$Q$, that entire segment is
  in or on~$Q$, a contradiction.  
  Thus the shortest path between $p$ and $q$ on the surface of $Q$ is the
  same as the shortest 3D path between $p$ and $q$ avoiding obstacle~$Q$.
  Because the interior of $Q$ is a subset of the interior of~$Q'$,
  the set of 3D paths from $p$ to $q$ avoiding obstacle $Q'$ is a subset
  of the set of 3D paths from $p$ to $q$ avoiding obstacle~$Q$,
  so the shortest path in the first set is at least as long as the
  shortest path in the second set.
\end{proof}

Next we consider the behavior of individual shortest paths during Tree-Unroll,
which degenerates to the Path-Unroll algorithm of
Section~\ref{Blooming a Refinement of Any Unfolding}:

\begin{lemma}\label{lem:tree-path-unroll}
  Suppose a shortest path on a convex polyhedron $Q$ is the
  3D polygonal path $P = \langle p_0, p_1, \dots, p_k \rangle$,
  with segment $p_{i-1} p_i$ contained in face $f_i$ of~$Q$
  and point $p_i$ on an edge of $Q$ for $0 < i < k$.
  Consider the motion of the polygonal path $P$ induced by
  running the Path-Unroll algorithm on the faces $f_1, f_2, \dots, f_k$.
  Then the path $P$ does not hit the plane containing $f_k$
  (except at $p_{k-1} p_k$) until the very end of the motion,
  when the whole path $P$ hits that plane.
\end{lemma}

\begin{proof}
  Let $h_i$ be the plane containing face $f_i$,
  and let $h_i^+$ be the closed halfspace bounded by $h_i$ and containing~$Q$.
  We define the \emph{grown polyhedron}
  $Q[P] = h_1^+ \cap h_2^+ \cap \cdots \cap h_k^+$
  to be the intersection of these halfspaces (which may be unbounded).
  The interior of $Q[P]$ contains the interior of~$Q$,
  and the endpoints of $P$ are on the boundary of $Q[P]$,
  so by Lemma~\ref{growing}, $P$~remains a shortest path in $Q[P]$.
  The heart of our proof will be maintaining this invariant
  throughout the motion.

  Now consider executing the first step of Path-Unroll,
  that is, rotating the first face $f_1$ to open the dihedral angle
  between $f_1$ and~$f_2$.
  As we perform the rotation, we keep the first edge $p'_0 p_1$ of the
  modified path $P'$ on the rotating face $f'_1$ and we maintain the
  grown polyhedron $Q[P']$.  Below we prove the following invariant:

  \begin{claim}\label{clm:stillshortest}
    The modified path $P'$ remains a shortest path on the grown
    polyhedron $Q[P']$, as long as $P'$ remains on the surface of $Q[P']$.
  \end{claim}

  Assuming this claim for now, one problem can still arise:
  if $p'_0$ reaches some plane $h_i$ defining one of the other faces $i > 2$
  (and thus enters that face),
  then we cannot continue rotating $f'_1$ (leaving $f_i$ behind)
  while keeping $P'$ on the boundary of $Q[P']$.
  But if this happens, then $P'$ could not actually be a shortest path,
  because there would be a shortcut $p'_0 p_i$ that lies within~$h_i$,
  instead of following the path $P'$ between $p'_0$ and~$p_i$
  (which is nonplanar because $i > 2$),
  contradicting the claim that $P'$ is shortest.

  Therefore the motion works all the way to when the dihedral angle between
  $f_1$ and $f_2$ is flat, reducing the number of faces in $Q[P']$.
  By induction, we can continue all the way and flatten the entire path, and
  by the previous argument, only at the end do we hit the final plane~$h_k$.
\end{proof}

It remains to prove Claim~\ref{clm:stillshortest}:

\medskip

\begin{proof}[of Claim~\ref{clm:stillshortest}]
  Suppose for contradiction that there is a shorter path $\tilde P$
  between $p'_0$ and $p_k$ on the grown polyhedron $Q[P']$;
  refer to Figure~\ref{trapezoid claim}.
  Let $\tilde p$ be the last point along the path $\tilde P$ 
  that touches the plane~$h_1$.  Such a point exists because $p'_0$ and
  $p_k$ are on opposite sides of~$h_1$, given that $p_k$ lies on the
  polyhedron $Q[P]$ and $p'_0$ has been rotated from a point $p_0$ on $h_1$
  by less than $180^\circ$ around a line $h_1 \cap h_2$ on~$h_1$.
  Because $\tilde p$ lies on the plane $h_1$ and within the halfspaces
  $h_2^+, h_3^+, \dots, h_k^+$ (being on the surface of~$Q[P']$),
  $\tilde p$~is on the surface of~$Q[P]$.
  Thus the segment $p_0 \tilde p$ lies on the original face~$f_1$.
  Similarly, the subpath $\tilde P[\tilde p,p_k]$ of $\tilde P$ from $\tilde p$
  to $p_k$ lies on the surface of $Q[P]$, because it lies on $Q[P']$
  and remains in the halfspace~$h_1^+$.
  Because $P$ is a shortest path from $p_0$ to $p_k$ on~$Q$, its length
  $|P| = |p_0 p_1| + |P[p_1 p_k]|$ is at most the length of any particular
  path, namely, $|p_0 \tilde p| + |\tilde P[\tilde p,p_k]|$.
  
  We now show that $|p'_0 \tilde p| \geq |p_0 \tilde p|$,
  which implies that $|\tilde P| = |p'_0 \tilde p| + |\tilde P[\tilde p,p_k]|
  \geq |p_0 \tilde p| + |\tilde P[\tilde p,p_k]| \geq |P| = |P'|$,
  contradicting that $\tilde P$ was shorter than~$P'$.
  Consider the triangle $p_0 p_1 \tilde p$, which lies in the plane $h_1$,
  and the rotated version of that triangle in $h'_1$, $p'_0 p_1 \tilde p'$.
  The quadrilateral $\tilde p p_0 p'_0 \tilde p'$ is an isosceles trapezoid,
  because $|p_0 \tilde p| = |p'_0 \tilde p'|$ and
  $p_0 p'_0$ is parallel to $\tilde p \tilde p'$
  (both being perpendicular to the bisecting plane of $h_1$ and~$h'_1$).
  Therefore the diagonal $\tilde p p'_0$ is longer than the equal sides,
  proving the last claim.
\end{proof}

\begin{figure}
  \centering
  \includegraphics{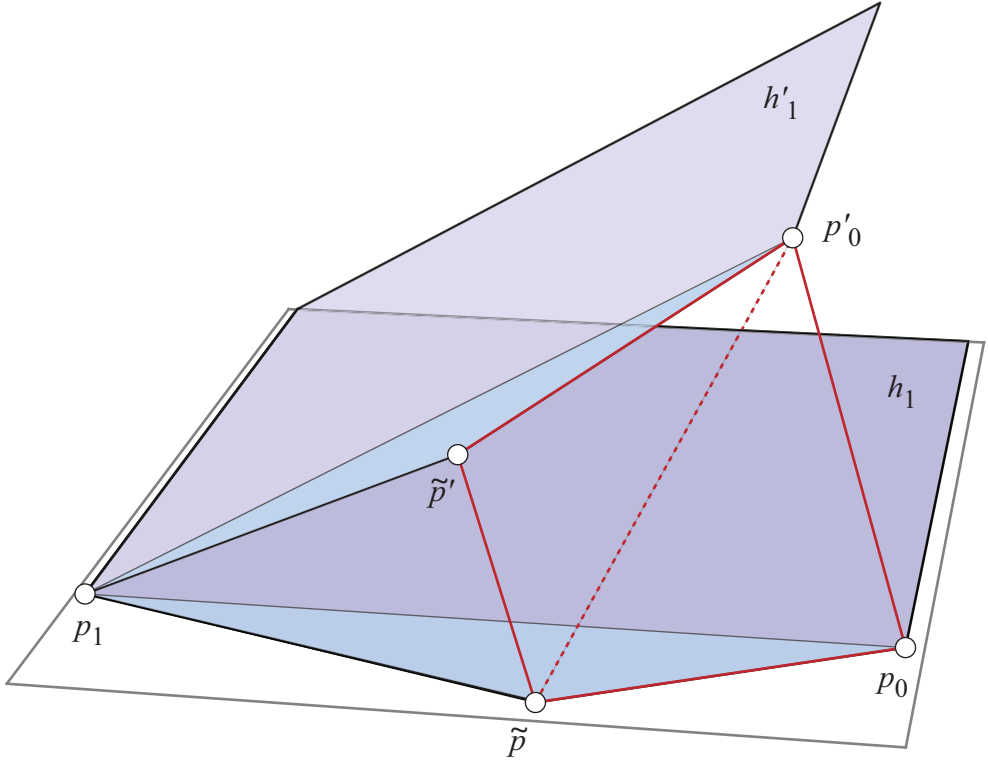}
  \caption{Proof of Claim~\ref{clm:stillshortest}}
  \label{trapezoid claim}
\end{figure}

Recasting Lemma~\ref{lem:tree-path-unroll} into the context of Tree-Unroll,
we obtain the following:

\begin{lemma}\label{lem-subtree}
  During the recursion of Tree-Unroll on some face $f$ of the tree~$T$,
  the subtree $T_f$ of faces in $T$ rooted at $f$ does not hit the plane
  containing the parent $\hat f$ of $f$ in $T$ until the algorithm has
  visited all of~$T_f$, at which time the entire subtree $T_f$
  simultaneously touches that plane.
\end{lemma}

\begin{proof}
  Suppose for contradiction that, during the unfolding of~$T_f$,
  some point $p$ of $T_f$ touches the plane containing~$\hat f$.
  By the definition of the source unfolding, the shortest path
  from the source $s$ to $p$ passes through face $\hat f$, say
  at some point $q$, and through face~$f$.
  Then, because Tree-Unroll unfolds that shortest path one face at a time,
  from $p$ to~$q$,
  Lemma~\ref{lem:tree-path-unroll} implies that the entire shortest path from $q$
  to $p$ has already been unfolded onto the plane containing~$\hat f$
  (given that $p$ is already there).
  But then face $f$ has also been unfolded onto the plane containing $\hat f$,
  which happens only when the algorithm has completed the recursion on~$f$.
\end{proof}

Finally we can prove the desired correctness of Tree-Unroll:

\begin{theorem}
  The Tree-Unroll algorithm causes no intersection during the blooming motion,
  for any source unfolding of any convex polyhedron.
\end{theorem}

\begin{proof}
  Each step of the algorithm unfolds the faces of the subtree $T_f$
  rooted at some face $f$ onto the plane containing the parent $\hat f$ of~$f$.
  During this motion, all other faces are either on the polyhedron~$Q$,
  or on planes containing the faces of the ancestors of $f$ in~$T$.
  Because the development of the subtree $T_f$ is a subset of the development
  of the entire source unfolding, the faces of the subtree $T_f$
  do not intersect each other in the plane containing~$f$.
  Because the faces of $T_f$ are in a supporting plane of the
  polyhedron~$Q$, those faces do not intersect any of the faces still on~$Q$
  that are not coplanar with~$\hat f$.
  More precisely, the plane containing the subtree $T_f$ intersects $Q$
  at the polyhedron face containing $f$ at the beginning of the motion,
  intersects $Q$ at the edge between $f$ and $\hat f$ during the motion,
  and intersects $Q$ at the polyhedron face containing $\hat f$
  at the end of the motion.
  At the end, when $T_f$ and $\hat f$ become coplanar,
  we know from the definition of the source unfolding
  that there is a shortest path from any point of the developed $T_f$ 
  to a point of $\hat f$ that crosses the edge between faces $f$ and~$\hat f$.
  Thus, the developed $T_f$ lies on the opposite side of the line containing
  this edge compared to the entire (convex) polyhedron face containing~$\hat f$,
  so $\hat f$ and $T_f$ do not intersect in their shared plane.
  Finally, by Lemma~\ref{lem-subtree}, the faces of $T_f$ do not intersect
  any of the planes containing the faces of the ancestors of~$f$,
  until the end of the motion when they all hit the plane containing~$\hat f$,
  and the previous argument applies.
  Therefore the blooming motion causes no intersection. 
\end{proof}

\section{Open Problems}
\label{Open Problems}

Cauchy's Arm Lemma states that opening the angles of a planar open chain
initially in convex position causes the two endpoints to get farther away.
This claim generalizes to 3D motions of an initially planar and convex
chain that only open the joint angles \cite{O'Rourke-2000-Cauchy}.
Our continuous blooming of the source unfolding suggests the following
related problem, phrased in terms of instantaneous motions like Cauchy's Arm
Lemma (but which would immediately imply the same about continuous motions):

\begin{open} \label{polycauchy}
  Consider the chain of faces visited by a shortest path on the surface
  of an arbitrary convex polyhedron.  If we open each
  dihedral angle between consecutive faces, and transform the edges of the
  shortest path along with their containing faces, does the resulting path
  always avoid self-intersection?
\end{open}

Indeed, we might wonder whether every dihedral-monotonic
blooming of the source unfolding avoids intersection.
This problem is equivalent to the following generalization of
Open Problem~\ref{polycauchy}.

\begin{open}
  Consider two shortest paths from a common point $s$ to points $t$ and $t'$
  on the surface of an arbitrary convex polyhedron, and
  consider the two chains of faces visited by these shortest paths.
  If we open each dihedral angle between consecutive faces
  in each chain, and transform the edges of the two shortest paths
  along with their containing faces,
  do the resulting paths always avoid intersecting each other and themselves?
\end{open}

We were unable to resolve either open problem using our techniques,
but it remains an intriguing question whether these analogs of
Cauchy's Arm Lemma underlie our positive results.%
\footnote{On the other hand, it is known that the faces of a convex polyhedron
  intersected by a shortest path can overlap, even when fully developed;
  the example is essentially \cite[Fig.~24.20, p.~375]{Demaine-O'Rourke-2007}
  (personal communication with G\"unter Rote).}

\section*{Acknowledgments}

This work was initiated at the 19th Bellairs Winter Workshop on Computational
Geometry held in Holetown, Barbados from January 30 to February 6, 2004.
We thank the other participants from that workshop---Greg Aloupis,
Prosenjit Bose, Mirela Damian, Vida Dujmovi\'c, Ferran Hurtado,
Erin McLeish, Henk Meijer, Pat Morin, Ileana Streinu, Perouz Taslakian,
Godfried Toussaint, Sue Whitesides, David Wood---for helpful discussions
and for providing a productive work environment.

\bibliography{blooming}

\newcommand{\etalchar}[1]{$^{#1}$}
\begin{thebibliography}{AHMS96}

\bibitem[AHMS96]{Arkin-Held-Mitchell-Skiena-1996}
Esther~M. Arkin, Martin Held, Joseph S.~B. Mitchell, and Steven~S. Skiena.
\newblock Hamiltonian triangulations for fast rendering.
\newblock {\em The Visual Computer}, 12(9):429--444, 1996.

\bibitem[AO92]{Aronov-O'Rourke-1992}
Boris Aronov and Joseph O'Rourke.
\newblock Nonoverlap of the star unfolding.
\newblock {\em Discrete \& Computational Geometry}, 8(3):219--250, 1992.

\bibitem[BDE{\etalchar{+}}03]{Bern-Demaine-Eppstein-Kuo-Mantler-Snoeyink-2003}
Marshall Bern, Erik~D. Demaine, David Eppstein, Eric Kuo, Andrea Mantler, and
  Jack Snoeyink.
\newblock Ununfoldable polyhedra with convex faces.
\newblock {\em Computational Geometry: Theory and Applications}, 24(2):51--62,
  February 2003.

\bibitem[BLS05]{Biedl-Lubiw-Sun-2005}
Therese Biedl, Anna Lubiw, and Julie Sun.
\newblock When can a net fold to a polyhedron?
\newblock {\em Computational Geometry: Theory and Applications},
  31(3):207--218, June 2005.

\bibitem[DFO07]{Damian-Flatland-O'Rourke-2007-epsilon}
Mirela Damian, Robin Flatland, and Joseph O'Rourke.
\newblock Epsilon-unfolding orthogonal polyhedra.
\newblock {\em Graphs and Combinatorics}, 23 (Supplement):179--194, 2007.

\bibitem[DO07]{Demaine-O'Rourke-2007}
Erik~D. Demaine and Joseph O'Rourke.
\newblock {\em Geometric Folding Algorithms: Linkages, Origami, Polyhedra}.
\newblock Cambridge University Press, July 2007.

\bibitem[IOV08]{Itoh-O'Rourke-Vilcu-2008}
Jin{-}ichi Itoh, Joseph O'Rourke, and Costin V{\^i}lcu.
\newblock Unfolding convex polyhedra via quasigeodesic star unfoldings.
\newblock Technical Report 091, Smith College, December 2008.
\newblock arXiv:0821.2257v1.

\bibitem[IOV09]{Itoh-O'Rourke-Vilcu-2009}
Jin{-}ichi Itoh, Joseph O'Rourke, and Costin V{\^i}ilcu.
\newblock Source unfoldings of convex polyhedra with respect to certain closed
  polygonal curves.
\newblock In {\em Abstracts from the 25th European Workshop on Computational
  Geometry}, Brussels, Belgium, March 2009.

\bibitem[MP08]{Miller-Pak-2008}
Ezra Miller and Igor Pak.
\newblock Metric combinatorics of convex polyhedra: cut loci and nonoverlapping
  unfoldings.
\newblock {\em Discrete \& Computational Geometry}, 39(1--3):339--388, March
  2008.
\newblock First manuscript in 2003.

\bibitem[O'R00]{O'Rourke-2000-Cauchy}
Joseph O'Rourke.
\newblock An extension of {C}auchy's arm lemma with application to curve
  development.
\newblock In J.~Akiyama, M.~Kano, and M.~Urabe, editors, {\em Revised Papers
  from the Japan Conference on Discrete and Computational Geometry}, volume
  2098 of {\em Lecture Notes in Computer Science}, pages 280--291, Tokyo,
  Japan, November 2000.

\bibitem[O'R08]{O'Rourke-2008-orthosurvey}
Joseph O'Rourke.
\newblock Unfolding orthogonal polyhedra.
\newblock In J.~E. Goodman, J.~Pach, and R.~Pollack, editors, {\em Surveys on
  Discrete and Computational Geometry: Twenty Years Later}, pages 231--255.
  American Mathematical Society, 2008.

\bibitem[PP09]{Pak-Pinchasi-2009}
Igor Pak and Rom Pinchasi.
\newblock How to cut out a convex polyhedron, 2009.
\newblock Manuscript.

\bibitem[SS86]{Sharir-Schorr-1986}
Micha Sharir and Amir Schorr.
\newblock On shortest paths in polyhedral spaces.
\newblock {\em SIAM Journal on Computing}, 15(1):193--215, February 1986.

\end{thebibliography}
\bibliographystyle{alpha}

\end{document}